\documentclass[12pt,letter,english]{article}

\usepackage[T1]{fontenc}
\usepackage[latin9]{inputenc}
\usepackage[table]{xcolor}
\usepackage{float}
\usepackage{amsmath}
\usepackage{amssymb}
\usepackage{babel}
\usepackage{bbm}
\usepackage{mathtools}
\usepackage{tikz}
\usepackage{dsfont}
\usepackage{amsthm}
\usepackage{enumerate}
\usepackage{cite}
\usepackage{caption}
\usepackage{subcaption}

\newcommand{\ubar}[1]{\underline{#1}}

\newtheorem{theorem}{Theorem}

\newtheorem{lemma}[theorem]{Lemma}
\newtheorem{corollary}[theorem]{Corollary}

\newtheorem{example}[theorem]{Example}

\newtheorem{remark}[theorem]{Remark}

\usetikzlibrary{positioning}
\DeclarePairedDelimiter\abs{\lvert}{\rvert}

\newcommand{\labeledunderbrace}[2]{
	\underset{#1}{\underbrace{#2}}
}

\def\p{\epsilon}
\def\pb{\bar{\p}}
\def\A{\ubar{A}}

\allowdisplaybreaks

\begin{document}

\title{Minimum Guesswork with an Unreliable Oracle}
\author{Natan Ardimanov, Ofer Shayevitz and Itzhak Tamo 
\thanks{The authors are with the Department of EE -- Systems, Tel Aviv University, Tel Aviv, Israel, emails \{natana@mail.tau.ac.il, ofersha@eng.tau.ac.il, zactamo@gmail.com\}. The work of I. Tamo and N. Ardimanov was supported by an ISF grant 1030/15 and an NSF-BSF grant 2015814. The work of O. Shayevitz was supported by an ERC grant no. 639573, ISF grant no. 1367/14, and ISF grant 1495/18.} \thanks{Parts of this work have been presented at the international symposium on information theory (ISIT) in Vail, CO, June 2018.}}
\date{}

\maketitle

\begin{abstract}
We study a guessing game where Alice holds a discrete random variable $X$, and Bob tries to sequentially guess its value. Before the game begins, Bob can obtain side-information about $X$ by asking an oracle, Carole, any binary question of his choosing. Carole's answer is however unreliable, and is incorrect with probability $\p$. We show that Bob should always ask Carole whether the index of $X$ is odd or even with respect to a descending order of probabilities -- this question simultaneously minimizes all the guessing moments for any value of $\p$. In particular, this result settles a conjecture of Burin and Shayevitz. We further consider a more general setup where Bob can ask a multiple-choice $M$-ary question, and then observe Carole's answer through a noisy channel. When the channel is completely symmetric, i.e., when Carole decides whether to lie regardless of Bob's question and has no preference when she lies, a similar question about the ordered index of $X$ (modulo $M$) is optimal. Interestingly however, the problem of testing whether a given question is optimal appears to be generally difficult in other symmetric channels. We provide supporting evidence for this difficulty, by showing that a core property required in our proofs becomes NP-hard to test in the general $M$-ary case. We establish this hardness result via a reduction from the problem of testing whether a system of modular difference disequations has a solution, which we prove to be NP-hard for $M\geq 3$. 

\end{abstract}

\section{Introduction and Main Result}\label{sec:intro}

Consider the classical guessing game played by Alice and Bob. Alice holds a discrete random variable (r.v.) $X$  distributed over $[N]\triangleq\{1,2,\ldots,N\}$ with a probability mass function  $p(x)$. Without loss of generality we assume below that the probabilities are in descending order, i.e., $p(1) \geq p(2) \geq\cdots\geq p(N)$. Bob would like to guess $X$ as quickly as possible. To that end, he is allowed to guess one symbol at a time, namely to ask Alice questions of the form ``is $X = x$''. Alice answers truthfully and the game terminates as soon as Bob guesses correctly. It is easy to check that to minimize his expected guessing time, Bob should guess the symbols in a decreasing order of probability. The minimal possible expected guessing time is therefore
\begin{align}
  G(X) \triangleq \sum_{k\in[N]} k\cdot p(k).
\end{align}
We note that $G(X)$ can be thought of as a measure of complexity of the random variable $X$. For example, in cryptographic settings, $G(X)$ measures the expected time for a brute-force attack on a random password with distribution $P_X$, and in communication, it can be related to the computational complexity of sequential decoding and the cut-off rate of the binary symmetric channel. As such, this quantity and its many variations have been extensively studied in the literature; see Section~\ref{sec:related_work} below for a short non-comprehensive survey.
%

In this paper, we consider the problem of guessing with the help of an unreliable oracle, recently introduced and studied by Burin and Shayevitz~\cite{burin2018reducing}. In this setup, before the game begins, Bob can reach out to an Oracle, Carole, who knows the true realization of $X$, and ask her any yes/no question that he believes would best help him in his subsequent guessing. This setup introduces a new dimension into the classical guessing game, as Bob is now allowed to actively gather some (noisy) information about the realization of $X$, incorporating his knowledge of the distribution. In a cryptographic setting, this can be thought of as a simplified toy model striving to capture the effect of actively collecting a limited amount of exogenous information that could help reduce a brute-force password attack (e.g., previous passwords, data from social media, etc.). In a communication setting, one can think of Alice and Carole as a single entity, and then interpret Carole's reply as a single information bit about $X$ that is sent over a binary symmetric channel to Bob, and then obtained by Alice/Carole via noiseless feedback; indeed, as shown in~~\cite{burin2018reducing}, allowing Bob to ask multiple binary questions is intimately related to the cut-off rate of the binary symmetric channels with feedback, and the question we ask here pertains to optimal feedback strategies when the decoder's performance is measured by its decoding time complexity.

Precisely, in our setup Bob can choose any subset $A\subseteq [N]$ and ask Carole whether $X$ is in $A$ or in $\bar{A} \triangleq [N]\setminus A$. Below we informally refer to the set $A$ as a {\em partition} (of $[N]$). Carole is known to lie with probability $\p$, i.e., Bob obtains the answer
	\begin{align}
		Y_A\triangleq \mathbbm{1}(X\in A)\oplus V,
	\end{align}
where $V\sim\text{Bernoulli}(\p)$ is independent of $X$. What is the best question, namely the best partition $A$, for Bob to choose so that given Carole's answer his expected guessing time would be minimized? If Carole is always truthful ($\p=0$), and recalling we have assumed that the symbols are ordered in a decreasing order of probabilities, it is simple to verify that Bob's best strategy is asking whether $X$ is even or odd, i.e., using  the partition $A_{\mathrm{ZZ}}=\{k\in[N] : k \text{ odd}\}$, which we refer to below as the {\em zigzag partition}. To see this, note first that given Carole's noiseless answer, the optimal guessing order is by decreasing order of probability over the symbols in either $A$ or $\bar A$. Thus, the guessing game will now terminate after $k$ steps with probability that is equal to the $k$th largest probability in $A$ plus the $k$th largest probability in $\bar A$. It is thus clear that to minimize the expected guessing time, Bob must make sure that symbols $1$ and $2$ are not both in $A$, and then more generally, that symbols $2k-1$ and $2k$ are not both in $A$. This property is guaranteed by the zigzag partition (as well as by $2^{\lceil N\rceil /2}$ other partitions).

What should the partition $A$ be in the general case? This turns out to be quite more difficult to determine, since the posterior distribution of $X$ given Carole's noise answer involves all the symbols, whose order is generally complicated. Let $G_A(X)$ be the optimal expected guessing time given Carole's noisy answer $Y_A$. In~\cite{burin2018reducing}, the authors reduced the problem of finding the partition that minimizes $G_A(X)$ to a max-cut problem in a certain weighted graph, and then used quadratic relaxation to prove that the zigzag partition is almost optimal, up to a small constant independent of the distribution and the alphabet size. 
\begin{theorem}[\cite{burin2018reducing}]
  For any r.v. $X$ and lying probability $\p$,
  \begin{align}\label{eq:Gbound}
    G_{A_{\mathrm{ZZ}}}(X) \leq \min_{A\subseteq[N]}G_A(X) + \frac{|1-2\p|}{4}.
  \end{align}
\end{theorem}
In addition, it was conjectured in~\cite{burin2018reducing} that the excess term in~\eqref{eq:Gbound} is an artifact of the proof, and that zigzag is in fact exactly optimal. In this paper, we prove this conjecture in a stronger sense, using an entirely different technique. For any function $f:[N]\to \mathbb{R}$, let $G^f_A(X)$ be the minimal expected value of $f(\textit{guessing time})$ given $Y_A$. We show:
\begin{theorem}\label{thrm:main}
  For any r.v. $X$, lying probability $\p$, and nondecreasing function $f:[N]\to \mathbb{R}$, 
  \begin{align}
    G^f_{A_{\mathrm{ZZ}}}(X) = \min_{A\subseteq [N]}G^f_A(X).
  \end{align}
\end{theorem}
Theorem~\ref{thrm:main}, proved in Sections~\ref{sec:lem} and~\ref{sec:main}, implies in particular that for any r.v.~$X$, the zigzag guessing strategy simultaneously minimizes all the positive guessing time moments for any noise level $\p$, and more specifically the conjecture in~\cite{burin2018reducing} follows by setting $f(k)=k$. It should be noted that the method of~\cite{burin2018reducing} cannot be extended to non-linear functions $f$, as the max-cut relaxation strongly hinges on the linearity of the expected guessing time. Our approach in this paper is based on a markedly different relaxation. Suppose that Bob wants to guess both $X$ and the noise $V$ \textit{simultaneously}, without access to any information, but is allowed to make \textit{two guesses} in each round. It is not difficult to see that any guessing strategy for the original problem can be converted to a guessing strategy for the relaxed problem with the exact same expected number of rounds; indeed, since $V$ is a function of $X$ and $Y_A$, at the $k$th round of the relaxed game Bob can simply submit his $k$th guess for $(X,V)$ under both $Y_A=0$ and $Y_A=1$. Therefore, the expected number of rounds in the relaxed game yields a lower bound on $G_A(X)$ for \textit{any} choice of $A$ (and a similar claim applies to $G_A^f(X)$). Surprisingly, we show that this lower bound is always achievable. Furthermore, there can generally be many optimal partitions that achieve this lower bound, and we concisely characterize and count all of them using graph-theoretic quantities. We then show that the zigzag partition is universal, in the sense that it is always a member of the set of optimal solutions. We note that given this result, it is straightforward to compute $\min_{A\subseteq [N]}G^f_A(X)$ in $O(N\log{N})$ time (e.g. using eq.~\eqref{eq:posterior_terms_eq}, or better yet, using Lemma~\ref{lem:gopt}). However, due to the combinatorial nature of the problem, it seems difficult to provide a tight analytical characterization of this quantity. Nevertheless, in Remark~\ref{rem:H_bound_Mary} we provide a lower bound as a function of the entropies of $V$ and $X$.

In Section~\ref{sec:multiple_choice}, we further discuss a natural extension of our setup to the case where Bob can ask a multiple-choice $M$-ary question, and where Carole's lies are modeled by a general modulo-additive channel. In this case,  a lower bound obtained via a relaxation similar to the binary case is not always tight, and in fact a (generalized) zigzag is not even necessarily optimal, unless the channel is fully symmetric, i.e., unless Carole has no preference when she lies. Moreover, our proof techniques do not directly extend to this case in a very strong sense: we show that testing for the validity of a core property required by our approach is NP-hard; this is established via a reduction from the problem of testing whether a system of modular difference disequations has a solution, which we prove to be NP-hard for $M\geq 3$. In Section~\ref{sec:discussion} we conclude with a discussion of some interesting directions trying to establish if testing whether the relaxation lower bound is tight can be done in polynomial-time in the alphabet size $N$. We further provide a polynomial-time algorithm for finding the optimal partition in the asymmetric binary case, for sufficiently small lying probabilities. 


\section{Related Work}\label{sec:related_work}
The problem of determining the value of a discrete r.v. $X$ by asking general binary questions is well studied in information theory, dating back to Shannon~\cite{shannon1948mathematical} and Huffman~\cite{huffman1952method}. It is well known that the Shannon entropy $H(X)$ is essentially the minimal number of questions required on average to describe a single realization of $X$, and is the exact limiting number of questions per instance (with high probability) required to describe multiple i.i.d. realizations of $X$. More recently, Massey~\cite{massey1994guessing} introduced a different notion of r.v. complexity, corresponding to the minimal number of guesses required on average in order to determine the value of $X$, referred to here as $G(X)$. Massey used the extremal maximum entropy properties of the geometric distribution to show that $G(X)$ grows at least exponentially with $H(X)$:{
\begin{align}\label{eq:massey_lb}
G(X) \geq \frac{2^{H(X)}}{\kappa}+1,
\end{align}
where $e + o(1) \leq \kappa \leq 4 $, where the $o(1)$ term vanishes as $G(X)$ grows large.} This bound it tight (with $\kappa = e)$ for geometrically distributed $X$. In a follow-up work, Arikan~\cite{arikan1996inequality} defined the notion of conditional guessing, and provided general lower and upper bounds on the $\rho$-th moment of the conditional guessing time of $X$ given some side-information $Y$, relating them to $H_{\frac{1}{1+\rho}}(X\mid Y)$, the Arimoto-R\'{e}nyi conditional entropy of order $\frac{1}{1+\rho}$. Arikan's upper bound without the conditioning was later tightened by Bozta{\c{s}}~\cite{boztas1997comments} for integer moments. In particular, when evaluated for a bivariate i.i.d. sequence $\{(X_k,Y_k)\}_{k=1}^n\stackrel{\mathrm{i.i.d}}{\sim} P_{XY}$ and $\rho=1$, Arikan's bounds imply that 
\begin{align*}
\lim_{n\to\infty}\frac{1}{n}\log{G(X^n\mid Y^n)} = H_{1/2}(X\mid Y), 
\end{align*}
with a similar result for general $\rho$. Continuing his previous work on the cutoff rate of single-user sequential decoding~\cite{arikan1988upper}, Arikan used the conditional guessing moment bounds to determine the cutoff rate of sequential decoding in multiple-access channels. We note that the problem consider in this paper can also be thought of as conditional guessing with side-information given by Carole's noisy answer; however, in the classical works on guessing with side information the distribution $P_{Y|X}$ is fixed, and one is interested in the behavior of the conditional guessing moments in the limit of multiple i.i.d. samples, whereas in our setup the main difficulty is to determine the best question Bob should ask Carole, which amounts to picking the best possible side-information from a given collection of options.

Arikan's work has been extended in many directions. Arikan and Merhav~\cite{arikan1998guessing} considered the case of guessing a possibly continuous r.v., where Bob's guess is considered correct if it is close enough to the true value w.r.t.~some distortion measure. They derived a single-letter variational expression for the exponent of the guessing moment as a function of the distortion level. The same authors then extended the discussion to a joint source-channel coding setup with a guessing decoder~\cite{arikan1998joint}, and to the wiretap channel setting with a guessing wiretapper~\cite{merhav1999shannon}. Arikan and Bozta{\c{s}} considered a one-sided lying variation of the guessing game~\cite{arikan2002guessing}, where Alice lies with some probability when she rejects Bob's guesses, but never lies when Bob guesses correctly. Sundaresan~\cite{sundaresan2007guessing} studied the case of universal guessing, where the distribution of $X$ is only known to belong to some family of distributions, and determined the associated penalty (redundancy) in the guessing exponent incurred by this uncertainty. When specialized to the case of an i.i.d.~distribution with an unknown marginal, his general results indicate that the redundancy term vanishes asymptotically, a fact that was already observed by Arikan and Merhav~\cite{arikan1998guessing}. Massey's guessing game has inspired a myriad of other works, tackling various other guessing setups and relations between guessing moments, entropy and compression, see e.g. ~\cite{pfister2004renyi,malone2004guesswork,yamamoto2011channel,hanawal2011guessing,sason2018improved,sason2018tight}, problems of guessing with a helper~\cite{Graczyk_Lapidoth,gr:17:thesis,weinberger2020guessing}, various universal guessing setups~\cite{sundaresan2006guessing,merhav2019universal}, and multi-agent guessing~\cite{christiansen2015multi,salamatian2017centralized}, as well as discussing the implications and applications of guessing in cryptographic settings, see e.g. ~\cite{arikan2008guessing,malone2012investigating,christiansen2013guessing,bracher2015guessing,yona2017effect}, among many others.

The guessing game considered in this paper allows Bob, the guesser, to ask a single general binary question to which he obtains a possibly incorrect answer, before proceeding with symbol-by-symbol guessing. The first phase of our setup is thus reminiscent of another game, known as the {\em R\'{e}nyi-Ulam game}~\cite{renyi1961problem,ulam1978adventures}. In this game, Bob is allowed to ask Carole multiple general binary questions (adaptively), to which he obtains possibly incorrect answers, and his goal is to identify $X$. The classical version of this game is adversarial: The number of questions Bob can ask as well as the maximum number of lies Carole can tell are given, and Bob needs to find $X$ with certainty (so there is no need to assume a distribution on $X$, only the cardinality $N$ matters). The problem is to determine, for a given set of parameters, whether Bob has a strategy to always win the game, see~\cite{pelc2002searching} for a comprehensive survey. In his PhD thesis, Berlekamp~\cite{berlekamp1964block} studied the properties of winnable games from the equivalent perspective of error correction with noiseless feedback. Specifically, he provided bounds on the asymptotic version of this problem, where the cardinality of $X$ grows exponentially as $N=2^{nR}$ and the maximum number of lies grows linearly as $n\p$, and where $n$ is the total number of questions. Berlekamp's bounds together with a result by Zigangirov~\cite{zigangirov1976number} provide a complete characterization of the relation between $\p$ and $R$, unlike the case where Bob needs to decide on his questions in advance, which is equivalent to the problem of finding the maximum growth rate of a binary error correcting code with minimum distance that scales linearly with the block length, a notorious open problem in coding theory.

In our setup the lies are random; the version of the R\'{e}nyi-Ulam game in which Carole lies with probability $\p$ and Bob needs to determine $X$ with high probability given Carole's answers, can essentially be thought of as the standard channel coding with noiseless feedback over a binary symmetric channel with crossover probability $\p$~\cite{horstein1963sequential,shayevitz2011optimal}. Going back to our guessing game, if we allow Bob to ask Carole multiple questions before he starts guessing $X$, then our setup can in fact be viewed as channel coding with noiseless feedback where instead of a small error probability we are interested in a small expected guessing time at the decoder. This problem is closely related to that of the {\em cutoff rate} of the binary symmetric channel with feedback~\cite{gallager1968information,arikan1988upper}, and a repeated zigzag partition has been used in~\cite{burin2018reducing} to provide a simple recursive scheme in the posterior matching spirit~\cite{shayevitz2011optimal} that attains this cutoff rate.

\section{Definitions and Basic Lemmas}\label{sec:lem}
Let us proceed more rigorously. We assume throughout without loss of generality that $0<\p<1/2$. First, recall the following well-known rearrangement lemma~\cite{hardy1952inequalities}. 
\begin{lemma}\label{lem:rearrange}
  Let $a_1,\ldots,a_N$ be a sequence of real numbers, and let $a_1^{\downarrow},\ldots, a_N^{\downarrow}$ be the same sequence ordered in descending order. Then 
  \begin{align}
    \sum_{k\in[N]}f(k)\cdot a_k^{\downarrow} \leq \sum_{k\in[N]}f(k)\cdot a_k,
  \end{align}
for any nondecreasing function $f: [N]\to \mathbb{R}$. 
\end{lemma}
\begin{proof}
  If $\{a_k\}$ is in  descending order, we are done. Otherwise, there must exist a pair $a_i>a_j$ for $i>j$. Switching between them clearly reduces the sum. Iterating this procedure, we terminate at a descending order after a finite number of iterations.  
\end{proof}

In light of Lemma~\ref{lem:rearrange}, it should be clear that for any partition $A$ and any non-decreasing function $f$, Bob's optimal guessing strategy in terms of minimizing the expectation of $f$ applied to his guessing time given $Y_A$, is to guess in decreasing order of posterior probabilities. In fact, this strategy is optimal simultaneously for all non-decreasing functions $f$, and specifically minimizes all the positive guessing time moments (we note that this simple fact has been already observed by Arikan, and has also been used in one-shot lossless source coding problems, see e.g. ~\cite{szpankowski2011minimum,kosut2017asymptotics}). Therefore, the minimal expected value of $f$ given $Y_A$ is  
\begin{align}
  G^f_A(X) = \mathbb{E} \left(\sum_{k\in[N]} f(k)\cdot P_{X|Y_A}(x_k^{Y_A} \mid Y_A)\right), 
\end{align}
where $\{x_k^y\}_{k\in[N]}$ is a permutation of $[N]$ that pertains to the {\em posterior order given $Y_A=y$}, i.e., such that 
\begin{align}
  P_{X|Y_A}(x_k^y\mid y) \geq  P_{X|Y_A}(x_{k+1}^y\mid y).
\end{align}

We are interested in studying the optimal partition, i.e., one that minimizes $G^f_A(X)$ over $A\subseteq [N]$.   
Writing $\pb = 1-\p$, let us expand the expression for $G_A^f(X)$:
{\renewcommand{\arraystretch}{1.2}%
\begin{align}
    G^f_A(X) &= \sum_{y\in\{0,1\}} \sum_{k\in[N]} f(k)\cdot P_{X|Y_A}(x_k^y\mid y)\cdot P_{Y_A}(y)\\
    &= \sum_{y\in\{0,1\}} \sum_{k\in[N]} f(k)\cdot P_{X,Y_A}(x_k^{y},y)\\
    &= \sum_{k\in[N]} f(k)\cdot \left[P_{Y_A|X}(0|x_k^0)p(x_k^0) + P_{Y_A|X}(1|x_k^1)p(x_k^1)\right]\\
    &= \sum_{k\in[N]}f(k)\cdot 
    \left\{
      \begin{array}{lr}
        \p\,  p(x_k^0) + \pb\,  p(x_k^1) & x_k^0\in A, x_k^1\in A\\
        \p\,  p(x_k^0) + \p\,  p(x_k^1) & x_k^0\in A, x_k^1\in \bar{A}\\
        \pb\,  p(x_k^0) + \pb\,  p(x_k^1) & x_k^0\in \bar{A}, x_k^1\in A\\
        \pb\,  p(x_k^0) + \p\,  p(x_k^1) & x_k^0\in \bar{A}, x_k^1\in \bar{A}, 
      \end{array}
\right. \label{eq:posterior_terms_eq}
\end{align}}
where the $\p,\pb$ in~\eqref{eq:posterior_terms_eq} terms follow immediately from the definition of how Carole lies. The above sum includes the $2N$ terms 
\begin{align}
\Pi \triangleq \{\p p(k)\}_{k\in[N]} \cup \{\pb p(k)\}_{k\in[N]}, 
\end{align}

which we refer to below as \textit{posterior terms}. We emphasize that in general, some posterior terms can have the same numerical value, in which case one needs to think of $\Pi$ as a \textit{multiset}; in fact, it is perhaps more instructive to think of the $2N$ posterior terms in $\Pi$ as distinct objects that are merely associated with numerical values. Nevertheless, for brevity of exposition we will assume throughout our derivations, unless otherwise stated, that all the posterior terms in $\Pi$ have distinct numerical values. This incurs no loss of generality, since if this is not the case then we can always consider an arbitrarily small perturbation of the distribution that satisfies this, see~\cite[Lemma 1]{burin2018reducing} for details. All the claims made in the paper are valid verbatim or with the trivial adaptations in the case of non-distinct posterior terms.

In~\eqref{eq:posterior_terms_eq}, the posterior terms are divided into $N$ pairs, where each pair consists of probabilities multiplied by either $\p$ or $\pb$, and where the set $A$ determines the assignment of probabilities to their multipliers. This assignment cannot be arbitrary in general, and is constrained by the possible choices of $A$. If one forgets about $A$ and allows an unconstrained assignment, then by Lemma~\ref{lem:rearrange} the minimum is obtained when all $2N$ posterior terms are first ordered in decreasing order, and then divided into contiguous pairs according to this order. This choice of pairs might not be consistent however with any partition $A$, and we therefore refer to the resulting sum as the \textit{unconstrained minimum}. Note that this unconstrained minimum precisely corresponds to the optimal strategy for the relaxed game discussion in Section~\ref{sec:intro}, where Bob tries to guess both $X$ and the noise $V$ in two-guesses rounds. The remainder of the paper is mostly concerned with the various properties of possible pairs and their relation to the partition $A$. Specifically, it will be shown that there are in general multiple choices of $A$ that attain the unconstrained minimum, and that the zigzag partition $A_{\mathrm{ZZ}}$ is always one of them.

In what follows, a pair of posterior terms of the form $(\p p(k), \pb p(k))$ is referred to as {\em posterior-siblings}. Note that the set $\Pi$ can be naturally written as a disjoint union of two {\em posterior sets} $\Pi = \Pi^0_A\cup \Pi^1_A$, where 
\begin{align}
	\label{post_sets_def}
	\Pi^y_A &\triangleq \left\{P_{X,Y_A}(x_k^{y},y)\right\}_{k\in[N]},
\end{align}
collects the posterior terms corresponding to answer $y$ by Carole. The following simple fact relates the posterior-siblings to posterior sets induced by a partition.
\begin{lemma}\label{lemma_pi_part}
The posterior sets separate all the posterior-siblings, i.e., they never both belong to the same posterior set $\Pi^y_A$. Conversely, for any partition of $\Pi = \Pi^0\cup \Pi^1$ that separates all the posterior-siblings, there exists a unique partition $A$ such that $\Pi_A^0=\Pi^0$ and $\Pi_A^1=\Pi^1$. 
\end{lemma}
\begin{proof} 
The first direction follows immediately from the definition. For the converse, write 
  \begin{align}
    \label{post_sets_def}
    \Pi^1_A &= \left\{\pb p(k) : k\in A\right\} \cup \left\{\p p(k) : k\in \bar{A}\right\}, 
  \end{align}
Hence $\Pi^1_A = \Pi^1$ implies that 
  \begin{align}
    A \triangleq\left \{k\in [N] : \pb p(k) \in \Pi^1\right\}. 
  \end{align}
Since $\Pi^0,\Pi^1$ separate the posterior-siblings, we also have 
  \begin{align}
    \bar{A} = \left \{k\in [N] : \p p(k) \in \Pi^1\right\}, 
  \end{align}
and using~\eqref{post_sets_def} again we have that $\Pi_A^1 = \Pi^1$. It is easy to check that $\Pi_A^0 = \Pi^0$ as well. 
\end{proof}

Let $\pi^y_A:[N]\to \Pi_A^y$ be the bijection recording the descending order on $\Pi_A^y$, i.e.:
\begin{align}
  \label{post_perms_def}
  \pi^y_A(k) \geq  \pi^y_A(k+1).
\end{align}
This bijection is unique by our assumption that all the posterior terms in $\Pi$ are distinct. With this notation at hand, we can write   
\begin{align}
  \label{G_A_from_post_map}
  G^f_A(X) = \sum_{k\in[N]}f(k)\cdot \left[\pi^0_A(k) + \pi^1_A(k)\right].
\end{align}

\begin{example}
  Let $N=2$ and $p(1) = 0.8 > p(2) = 0.2$. Carole lies with probability $\p = 0.1$, and Bob chooses $A = \{1\}$.
  This choice generates the following $\Pi^y_A$ sets:
  \begin{align}
    &\Pi^0_A = \{\p p(1), \pb p(2)\}\\
    &\Pi^1_A = \{\pb p(1), \p p(2)\},
  \end{align}
  $\pi^y_A$ in this case is 
  \begin{align}
    &\pi^0_A(1) = \pb p(2) = (1-0.1)\cdot 0.2 = 0.18\\
    &\pi^0_A(2) = \p p(1) = 0.1\cdot 0.8 = 0.08\\
    &\pi^1_A(1) = \pb p(1) = (1-0.1)\cdot 0.8 = 0.72\\
    &\pi^1_A(2) = \p p(2) = 0.1\cdot 0.2 = 0.02. 
  \end{align}
  and 
  \begin{align}
  	G^f_A(X) &= f(1)\left[\pi^0_A(1) + \pi^1_A(1)\right] + f(2)\left[\pi^0_A(2) + \pi^1_A(2)\right] \\
  	         &= f(1)\left[\pb p(2) + \pb p(1)\right] + f(2)\left[\p p(1) + \p p(2)\right]\\
  	         & = 0.9\cdot f(1) + 0.1\cdot f(2).
  \end{align}

\end{example} 
        
A bijection $\sigma:[2N]\to \Pi$ is {\em induced by $A \subseteq [N]$}, if for  all $k\in [N]$
\begin{align}
  \label{feas_cond}
  \{\sigma(2k-1), \sigma(2k)\} = \{\pi^0_A(k), \pi^1_A(k)\}.
\end{align}
Below we refer to $\sigma(2k-1)$ and $\sigma(2k)$ as {\em $\sigma$-siblings}. Note that for this $\sigma$,  
\begin{align}\label{eq:gfsum_sigma}
  G^f_A(X) = \sum_{k\in[N]}f(k)\cdot \left[\sigma(2k-1) + \sigma(2k)\right]. 
\end{align}

It is not difficult to check (e.g., by counting) that not all bijections $\sigma$ are induced by a partition. Let us now characterize the bijections that are induced by some $A$. We say that $A$ and $\sigma$ are {\em posterior-respecting} if $\Pi^0_A$ and $\Pi^1_A$ separate all the $\sigma$-siblings, i.e., 
\begin{align}
  \{\sigma(2k-1), \sigma(2k)\} \not\subseteq \Pi^y_A.
\end{align}
for any $k\in[N], y\in\{0,1\}$. In the context of our problem, the posterior-respecting property guarantees that the $k$th odd/even pair of posterior terms w.r.t. the order $\sigma$, would be a feasible $k$th guess for Bob under a positive/negative response from Carole to the partition $A$, i.e., will contain posterior terms corresponding to both possible answers. This in itself is however not sufficient; recall that when Bob guesses optimally, he guesses in decreasing order of posterior probabilities. Hence, we further say that a set $A$ and $\sigma$ are {\em order-preserving} if the elements of both $\Pi_A^0$ and $\Pi_A^1$ are ordered within the bijection, i.e., $\sigma(i) > \sigma(j)$ whenever $\{\sigma(i),\sigma(j)\} \subseteq \Pi^{y}_A$ for $i < j$ and some $y$. Together, we have:
\begin{lemma}
    \label{lemma_feas_bi}
    $\sigma$ is induced by $A$ if and only if they are posterior-respecting and order-preserving.
\end{lemma}
	
\begin{proof}
  If $\sigma$ is induced by $A$, the claim follows trivially from definition. Suppose $A$ and $\sigma$ are posterior-respecting. Then for any $k\in[N]$ there exists $y_k\in\{0,1\}$ such that the $k$th $\sigma$-siblings are separated:
  \begin{align}
    \sigma(2k-1) \in \Pi^{y_k}_A,\ \sigma(2k) \in \Pi^{1-y_k}_A.
  \end{align}
  This enables us to define the bijections $\sigma^0:[N]\to \Pi_A^0$ and $\sigma^1:[N]\to \Pi_A^1$ by 
  \begin{align}
    &\sigma^{y_k}(k) \triangleq \sigma(2k-1) \\
    &\sigma^{1-y_k}(k) \triangleq \sigma(2k) .
  \end{align}
  If $A$ and $\sigma$ are also order-preserving then
  it must be that $\sigma^0(1) > \sigma^0(2) > \ldots > \sigma^0(N)$, which means that $\sigma^0$ is a bijection
  from $[N]$ to $ \Pi^0_A$ that agrees with the posterior order bijection $\pi^0_A$. Since $\pi^0_A$ is unique, we conclude that
  $\sigma^0 = \pi^0_A$. Similarly, $\sigma^1 = \pi^1_A$. We have thus obtained the following set equalities:
  \begin{align}
    \{\sigma(2k-1), \sigma(2k)\} &= \{\sigma^{y_k}(k), \sigma^{1-y_k}(k)\}\\
    &= \{\sigma^0(k), \sigma^1(k)\}\\
    &= \{\pi^0_A(k), \pi^1_A(k)\},
  \end{align}
  and hence $\sigma$ is induced by $A$. 		
\end{proof}

We will now show (in Lemma~\ref{lemma_graph} below) that although there are many bijections $\sigma$ that are not induced by any partition, the only obstacle is the order-preserving property; we can always find (often many) partitions $A$ that satisfy the posterior-respecting property. This, together with Lemma~\ref{lemma_feas_bi}, will imply (as later argued in Lemma~\ref{lem:gopt}) that there are partitions that induce the natural descending order of posterior terms, which in turn attains the unconstrained optimum in our guessing problem. To that end, we first need to introduce the graph $\mathcal{G}_\sigma$ induced by a bijection ${\sigma:[2N]\to \Pi}$. The vertex set of $\mathcal{G}_\sigma$ is the set $\Pi$ of posterior terms, and we draw an edge between any two vertices that are either posterior-siblings or $\sigma$-siblings. 
\begin{lemma}
  $\mathcal{G}_\sigma$ is a disjoint union of even cycles and isolated edges. 
\end{lemma}
\begin{proof}
By definition, the degree of each vertex $v$ is ${\deg(v)\in\{1,2\}}$.  If $\deg(v)=1$, and denoting its single adjacent vertex by $v'$, then $(v,v')$ are both posterior-siblings and $\sigma$-siblings, and hence $\deg(v')=1$. Thus the graph is a disjoint union of degree-$1$ vertices (i.e., isolated edges) and degree-$2$ vertices. The component of degree-$2$ vertices must be a disjoint union of cycles. Because both posterior-siblings must be in the same cycle, each cycle is of even length. 
\end{proof}
The following corollary is immediate. 
\begin{corollary}\label{cor:g_sigma}
$\mathcal{G}_\sigma$ is $2$-colorable, and the number of distinct colorings is $2^c$, where $c$ is the number of connected components of $\mathcal{G}_\sigma$. 
\end{corollary}
\nocite{bondy1976graph}

With this in hand, we can prove the following.  
\begin{lemma}
  \label{lemma_graph}
  For any bijection $\sigma:[2N]\to \Pi$ there exists a partition $A$ such that $A$ and $\sigma$ are posterior-respecting. 
  Moreover, the number of such partitions $A$ is $2^c$, where $c$ is the number of connected components of $\mathcal{G}_\sigma$.
\end{lemma}

\begin{proof}
  Fix some  $2$-coloring of $\mathcal{G}_\sigma$, which must exist by Corollary~\ref{cor:g_sigma}. Let $\Pi^0$ and $\Pi^1$ be the color classes associated with this coloring, which form a  partition of the vertex set $\Pi$ into two independent sets. Since posterior-siblings are connected by an edge, it follows that $\Pi^0$ and $\Pi^1$ separate all the posterior siblings. Thus according to Lemma~\ref{lemma_pi_part} there exists a set $A$ such that $\Pi^0 = \Pi^0_A$ and $\Pi^1 = \Pi^1_A$. Since $\sigma$-siblings are also connected by an edge, it follows that $\Pi_A^0$ and $\Pi_A^1$ separate all the $\sigma$-siblings. Hence, $A$ and $\sigma$ are posterior-respecting. Finally, any $2$-coloring clearly results in a distinct and unique $A$ satisfying the condition, hence in light of Corollary~\ref{cor:g_sigma} there are $2^c$ such partitions.  
\end{proof}

	\begin{example}\label{ex:G}
		Let $N = 4$ and define the following bijection $\sigma$:
		\begin{alignat*}{2}
			\sigma(2\cdot1-1) &= \pb  p(1)\qquad  &\sigma(2\cdot1) &= \pb  p(2)\\
			\sigma(2\cdot2-1) &= \p p(1) &\sigma(2\cdot2) &= \pb  p(3)\\
			\sigma(2\cdot3-1) &= \p p(2) &\sigma(2\cdot3) &= \p p(3)\\
			\sigma(2\cdot4-1) &= \pb p(4) &\sigma(2\cdot4) &= \p p(4).
		\end{alignat*}
		The corresponding graph $\mathcal{G}_\sigma$ and a legal $2$-coloring appears in Figure~\ref{fig:ex}. 
        \begin{figure}[t!]
        	\centering
        	\captionsetup{justification=centering}
			\begin{tikzpicture}
				\node[circle,label=above:$\pb  p(1)$,draw=black,fill,color=yellow,minimum size=0.5cm] (1) 	{};
				\node[circle,label=below:$\p p(1)$,draw=black,fill,color=red,minimum size=0.5cm] (2) [below = 1cm of 1]  {};
				\node[circle,label=below:$\pb  p(3)$,draw=black,fill,color=yellow,minimum size=0.5cm] (3) [right = 1.5cm of 2]	{};
				\node[circle,label=below:$\p p(3)$,draw=black,fill,color=red,minimum size=0.5cm] (4) [right = 1.5cm of 3]	{};
				\node[circle,label=above:$\p p(2)$,draw=black,fill,color=yellow,minimum size=0.5cm] (5) [above = 1cm of 4]	{};
				\node[circle,label=above:$\pb p(2)$,draw=black,fill,color=red,minimum size=0.5cm] (6) [left = 1.5cm of 5]	{};
				
				\node[circle,label=below:$\p p(4)$,draw=black,fill,color=red,minimum size=0.5cm] (7) [right = 1.5cm of 4]	{};
				\node[circle,label=above:$\pb p(4)$,draw=black,fill,color=yellow,minimum size=0.5cm] (8) [above = 1cm of 7]	{};
				
				\path[draw,thick]
				(1) edge node {} (2)
				(2) edge node {} (3)
				(3) edge node {} (4)
				(4) edge node {} (5)
				(5) edge node {} (6)
				(6) edge node {} (1)
				(7) edge node {} (8);
			\end{tikzpicture}
           	\caption{The graph $\mathcal{G}_\sigma$ and a $2$-coloring for Example~\ref{ex:G}}
           	\label{fig:ex}
		 \end{figure}
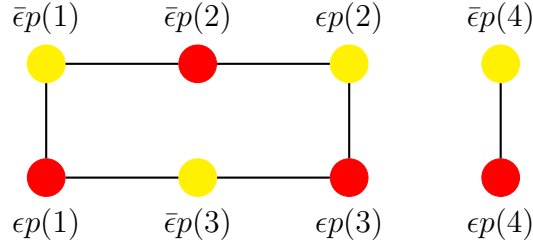
		If $\Pi^1$ is the set of the red nodes, then $\Pi^1=\Pi^1_A$ for $A=\{2\}$. Otherwise, If $\Pi^1$ is the set of the yellow nodes, then $\Pi^1=\Pi^1_A$ for  $A=\{1,3,4\}$. The number of legal 2-colorings is $2$ for each one of the connected components, in total $G_\sigma$ has $2\cdot 2 =4 $ legal 2-colorings, and each 2-coloring corresponds to different partition $A$.
	\end{example}

\section{Proof of Theorem~\ref{thrm:main}} \label{sec:main}
	\label{sec:zigzag_is_optimal}
	Define $\sigma^\downarrow: [2N] \to \Pi$ to be the unique bijection corresponding to the natural descending order on $\Pi$, i.e., such that   
	\begin{align}
		\label{pi_ord_prop}
		\sigma^\downarrow(1) > \sigma^\downarrow(2) > \cdots > \sigma^\downarrow(2N). 
	\end{align}
        
        In light of Lemma~\ref{lem:rearrange} and since~\eqref{eq:gfsum_sigma} holds for any permutation $\sigma$ induced by some partition $A$, we clearly have that 
        \begin{align}\label{eq:best_possible}          
          \min_{A\subseteq [N]}G^f_A(X) \geq  \sum_{k\in[N]}f(k)\cdot \left[\sigma^\downarrow(2k-1) + \sigma^\downarrow(2k)\right].
        \end{align}
        The right-hand-side of~\eqref{eq:best_possible} is an unconstrained minimum, since not all permutations are induced by a partition. Somewhat surprisingly, the permutation $\sigma^\downarrow$ that achieves the unconstrained minimum, is in fact always induced by some partition. 
	
\begin{lemma}
  \label{lem:gopt}
  For an optimal partition, it holds that 
  \begin{align} 
    \min_{A\subseteq [N]}G^f_A(X) = \sum_{k\in[N]}f(k)\cdot \left[\sigma^\downarrow(2k-1) + \sigma^\downarrow(2k)\right].
  \end{align}
  Moreover, the number of optimal partitions is $2^c$, where $c$ is the number of connected components of $\mathcal{G}_{\sigma^\downarrow}$.  
\end{lemma} 

\begin{proof}
  By Lemma \ref{lemma_graph}, there exists some partition $A^\downarrow$ such that $A^\downarrow$ and $\sigma^\downarrow$ are posterior-respecting. It is easy to see that $A^\downarrow$ and $\sigma^\downarrow$ are order-preserving; this in fact holds for any partition $A$ simply since $\sigma^\downarrow$ is ordered. Invoking Lemma~\ref{lemma_feas_bi}, $\sigma^\downarrow$ is induced by $A^\downarrow$ and the claim follows from~\eqref{eq:gfsum_sigma}. The number of optimal partitions now follows from Corollary~\ref{cor:g_sigma}. 
\end{proof}
\begin{remark}
  Note that when counting the number of optimal partitions, we are counting partitions and their complements, which essentially corresponds to the same solution. The number of truly distinct solutions is therefore $2^{c-1}. $ 
\end{remark}
\begin{remark}
  It is instructive to examine two extremal cases. First, consider the completely noisy case where $\p =1/2$, where clearly any partition is optimal (and useless). Indeed, in this case $\mathcal{G}_{\sigma^\downarrow}$ comprises a disjoint union $N$ isolated edges, simply since each pair of posterior siblings is also a pair of $\sigma^\downarrow$-siblings, hence there are exactly $N$ connected components yielding $2^N$ optimal partitions, which is the maximal possible. Note that the same conclusion remains valid when $\p < 1/2$ but is sufficiently close to $1/2$ such that $\p p(k) \geq \pb p(k+1)$ for all $k$. Now, on the other extreme, consider the noiseless case where $\p =0$ and assume for simplicity that $N$ is even (note that here, there are $N$ distinct posterior terms are associated with the same numerical value of zero). In this case, $\mathcal{G}_{\sigma^\downarrow}$ comprises a disjoint union of cycles of length $4$, each of the form $(\pb p(2k-1), \p p(2k-1), \p p(2k), \pb p(2k))$. Hence, the number of optimal partitions is given by $2^{N/2}$. Indeed, this makes sense: It is obvious that in this case the zigzag is an optimal partition, but we can also switch the order of any (odd, even) pair in the partition without sacrificing anything, yielding $2^{N/2}$ possible solutions. Note that the same conclusions remain valid in the almost-
  noiseless regime, when $\p$ is sufficiently small such that $\p p(1) \leq \pb p(N)$.
\end{remark}

We have seen that the unconstrained minimum can be attained, and that in general, there may be many partitions that attain it. But it is still unclear what these optimal partitions look like. Interestingly, we now show that the zigzag partition is always a member of the set of optimal partitions, which concludes the proof of Theorem~\ref{thrm:main}. To that end, it suffices to show the following: 
	\begin{lemma}
		\label{lemma_azz}
		$A_{ZZ}$ and $\sigma^\downarrow$ are posterior-respecting. 
	\end{lemma} 
\begin{proof}
  $A_{ZZ}$ partitions $\Pi$ into the following two posterior sets:
  \begin{align}
    \Pi^0_{A_{ZZ}} =& \{\p p(k) : k \text{ odd}\} \cup \{\pb p(k) : k\text{ even}\} \label{eq:zz_posterior1} \\
    \Pi^1_{A_{ZZ}} =& \{\pb p(k) : k \text{ odd}\} \cup \{\p p(k) : k\text{ even}\}\label{eq:zz_posterior2}.
  \end{align}
  To prove our claim, we need to show that this partition separates all the $\sigma^\downarrow$-siblings. To that end, we make a distinction between different types of $\sigma^\downarrow$-siblings:  
  \begin{enumerate}[(i)]
  \item The $\sigma^\downarrow$-siblings are of the form $\{\p p(i), \p p(j)\}$: Since the probabilities are descending order ($p(k) > p(k+1)$) and $\sigma^\downarrow$ also orders the posterior terms in descending order (cf.~\eqref{pi_ord_prop}), then it must be that $|j-i|=1$. Hence, we conclude that $j$ and $i$ must have different parities. In light of~\eqref{eq:zz_posterior1}-\eqref{eq:zz_posterior2}, it is clear that these $\sigma^\downarrow$-siblings cannot belong to the same posterior set.

\item The $\sigma^\downarrow$-siblings are of the form $\{\pb p(i), \pb p(j)\}$: This follows similarly to the previous case. 

\item The $\sigma^\downarrow$-siblings are of the form $\{\p p(i), \pb p(j)\}$: 
  Since the probabilities are descending order ($p(k) > p(k+1)$) and $\sigma^\downarrow$ also orders the posterior terms in descending order (cf.~\eqref{pi_ord_prop}), and $\p < 1/2$, it must hold that $i\leq j$. Let us count how many posterior terms are greater than both  $\{\p p(i), \pb p(j)\}$. These terms are exactly all the terms of the form $\{\pb p(k)\} _{k=1}^{j-1}$, $\{\p p(k)\}_{k=1}^{i-1}$, a total of exactly ${(j-1) + (i-1) = i+j-2}$ terms. This number must be even, since the $\sigma^\downarrow$-siblings come in pairs. Therefore, $i$ and $j$ must have the same parity and again, in light of~\eqref{eq:zz_posterior1}-\eqref{eq:zz_posterior2}, it is clear that these $\sigma^\downarrow$-siblings cannot belong to the same posterior set.
  \end{enumerate}
\end{proof}

This concludes the proof of our main result. A simple consequence is the following: 
\begin{corollary}\label{cor:zz_unique_cond}
  The zigzag partition is the unique optimal partition (up to complements) if and only  if $\mathcal{G}_{\sigma^\downarrow}$ is a cycle on $2N$ vertices. 
\end{corollary}

\begin{remark}
  The condition provided in Corollary~\ref{cor:zz_unique_cond} is simple to check, but is not very intuitive; specifically, it is not immediately clear which distributions and noise levels could satisfy it, if any. To show it is not vacuous and gain some insight, consider the following special case. Suppose that the distribution $p(k)$ and noise level $\p$ satisfy the following condition:
  \begin{align}
      \max_{k\in[N-2]} \frac{p(k+2)}{p(k)} < \p / \pb < \min_{k\in[N-1]} \frac{p(k+1)}{p(k)}.
  \end{align}
  It is then rather straightforward to check that $\mathcal{G}_{\sigma^\downarrow}$ is a cycle on $2N$ vertices, and hence the zigzag partition is the unique optimal partition. To see this, demonstrating for simplicity using $N=5$, note that this condition implies that the posterior terms satisfy 
  \begin{align}
      \pb p(1) > \pb p(2) > \p p(1) > \pb p(3) > \p p(2) > \pb p(4) > \p p(3) > \pb p(5) > \p p(4) > \p p(5)
  \end{align}
  which induces the $10$-cycle $(\pb p(1), \p p(1), \pb p(3), \p p(3) , \pb p(5), \p p(5), \p p(4), \pb p(4), \p p(2), \pb p(2))$. While for many distributions this condition cannot be satisfied by any noise level $\p$, there are cases where the condition holds, for any alphabet size $N$. Indeed, it is easy to see that this happens when the distribution is (truncated) geometric or approximately so, e.g., when $p(k) = \alpha \beta^k$ for some $\beta\in (0,1)$ and a suitable $\alpha$ and one picks any $\p$ such that $\beta^2 < \p / \pb < \beta$. 
\end{remark}

\section{Multiple-choice questions}\label{sec:multiple_choice}
A natural extension of the problem considered so far is a setup in which Bob, instead of asking a binary question, can ask a multiple-choice $M$-ary question for some $M>2$, i.e., can partition $[N]$ into $M$ sets $\{A^i\}_{i=0}^{M-1}$ and ask Carole which one contains $X$. In this case, instead of a binary symmetric channel, it is natural to assume that Carole's answer is corrupted by some $M$-ary symmetric modulo-additive channel. Explicitly,  writing $\A \triangleq \{A^i\}_{i=0}^{M-1}$ for short, and given that $X \in A^j$, Bob receives a noisy answer
\begin{align}
	Y_{\A} \triangleq j + V \mod M,
\end{align}
where $P_V(v) = \p_v$ for $v \in \{0, 1,\ldots,M-1\}$. In what follows, we generalize our previous definitions and show that it is computationally hard to test whether, for a general bijection $\sigma$, there exists a partition $\A$ such that $A$ and $\sigma$ are posterior-respecting; namely, we show that Lemma~\ref{lemma_graph} no longer holds for $M > 2$. This observation seems to suggest that it is perhaps also hard to test whether a partition $A$ is optimal. We conclude this section with an example showing that unlike in the binary case, for $M > 2$ the unconstrained optimum cannot be always achieved. It is nevertheless worth noting that in the special case where the channel is fully symmetric, i.e., where Carole has no preference when she lies, a generalized modulo-$M$ zigzag question achieves the unconstrained optimum; see a brief discussion in Section~\ref{sec:discussion}. 

Below, we refer to $\{\p_0 p(k),\p_1 p(k), \ldots, \p_{M-1} p(k)\}$ as posterior-siblings, and we let 
\begin{align}
	\Pi \triangleq \bigcup_{k\in[N]} \{\p_0 p(k),\p_1 p(k), \ldots, \p_{M-1} p(k)\},
\end{align}
be the set of all posterior-siblings. Following the same definition for posterior sets $\Pi^y_{\A}$,
\begin{align}
	\Pi^y_{\A} &\triangleq \left\{P_{X,Y_{\A}}(x_k^{y},y)\right\}_{k\in[N]},
\end{align}
and the corresponding bijections $\pi^y_{\A}$, we can write Bob's expected $f({\textrm{guessing time}})$ given Carole's answer as 
\begin{align}
	\label{G_A_from_post_map_mary}
	G^f_{\A}(X) = \sum_{k\in[N]}f(k)\cdot \sum_{y=0}^{M-1}\pi^y_{\A}(k).
\end{align}

\begin{remark}\label{rem:H_bound_Mary}
Before we proceed, it is worth noting that one can use Massey's bound~\eqref{eq:massey_lb} to obtain a lower bound on the expected guessing time $G_{\A}(X)$, i.e., where $f(k) = k$. We restrict ourselves to this case only for simplicity of exposition; one can derive similar lower bounds on $G^f_{\A}(X)$ using maximum entropy arguments, see e.g.~\cite{weinberger2020guessing}. To obtain the bound, note that $H(X|Y_{\A}) \geq H(X) + H(V) - \log M$. Applying Massey's bound~\eqref{eq:massey_lb} to the conditional distribution $P_{X|Y_{\A}}(\cdot |Y_{\A})$ and averaging over Carole's answers, we obtain
\begin{align}
    G_{\A}(X) &\geq \mathbb{E}_{Y_{\A}}\left(\frac{2^{H(P_{X|Y_{\A}}(\cdot |Y_{\A}))}}{\kappa}\right) + 1 \\ 
    & \geq \frac{2^{H(X|Y_{\A})}}{\kappa} + 1 \label{eq:lb_maxent_jensen}\\
    & = \frac{2^{H(X) + H(V)}}{\kappa M} + 1, 
\end{align}
where we have used Jensen's inequality in~\eqref{eq:lb_maxent_jensen}. 
\end{remark}

Given a partition $\{\Pi^i\}_{i=0}^{M-1}$ of $\Pi$, we say that it {\em cyclically-separates} the posterior-siblings $\{\p_0 p(k),\p_1 p(k), \ldots, \p_{M-1} p(k)\}$ if 
\begin{align}
	\label{cyc_sep_cond}
	\p_i p(k) \in \Pi^j \Longleftrightarrow \p_{i+1} p(k)\in \Pi^{j+1},
\end{align}
where indices are calculated modulo $M$.
\begin{lemma} \label{lemma_gen_pi_part}
	 (Generalization of Lemma~\ref{lemma_pi_part}) 
	For any $A$, the posterior sets cyclically-separate all the posterior-siblings. Conversely, for any partition $\{\Pi^i\}_{i=0}^{M-1}$ of $\Pi$ that cyclically-separates all the posterior-siblings, there exists a unique partition $\A$ such that $\Pi_{\A}^y = \Pi^y$.
\end{lemma}
\begin{proof}
	The first direction follows immediately from the definition. The converse is a trivial generalization of the converse in Lemma~\ref{lemma_pi_part}.
\end{proof}
        
Continuing generalization of previous definitions, a bijection $\sigma:[MN]\to \Pi$ is {\em induced by $\A$}, if for  all $k\in [N]$
\begin{align}
\label{gen_induc_cond}
\{\sigma(Mk), \sigma(Mk-1), \ldots, \sigma(M(k-1)+1)\} = \{\pi^0_{\A}(k), \pi^1_{\A}(k), \ldots, \pi^{M-1}_{\A}(k)\},
\end{align}
and for such $\sigma$
\begin{align}\label{eq:gen_gfsum_sigma}
G^f_A(X) = \sum_{k\in[N]}f(k)\cdot \sum_{i=0}^{M-1} \sigma(Mk-i). 
\end{align}
where $\{\sigma(Mk-i)\}_{i=0}^{M-1}$ are {\em $\sigma$-siblings} in the general case. $\A$ and $\sigma$ are {\em posterior-respecting} if 
$\{\Pi^y_{\A}\}_{y=0}^{M-1}$ separate (not necessary cyclically-separate)  all the $\sigma$-siblings, i.e., 
\begin{align}\label{gen_post_resp_cond}
\{\sigma(Mk-i), \sigma(Mk-j)\} \not\subseteq \Pi^y_{\A},
\end{align}
for any $k\in[N]$, $y, j, i\in\{0,1,\ldots,M-1\}$ and $i\ne j$. Completing the generalization, $A$ and $\sigma$ are {\em order-preserving} if 
the elements of $\Pi_{\A}^y$ are ordered within the bijection, i.e., $\sigma(i) > \sigma(j)$ whenever $\{\sigma(i),\sigma(j)\} \subseteq \Pi^{y}_{\A}$ for $i < j$ and some $y$. 
\begin{lemma} (Generalization of Lemma~\ref{lemma_feas_bi})
	\label{lemma_gen_feas_bi}
	$\sigma$ is induced by $\A$ if and only if they are posterior-respecting and order-preserving.
\end{lemma}
\begin{proof}
	If $\sigma$ is induced by $\A$, the claim follows trivially from definition. Suppose $\A$ and $\sigma$ are posterior-respecting. Then for any $k\in[N]$ there exists a permutation $\gamma_k$ of $\{0,1,\ldots,M-1\}$ such that the $k$th $\sigma$-siblings separated as follows
	\begin{align}
	\sigma(Mk-i) \in \Pi^{\gamma_k(i)}_{\A}.
	\end{align}
	This enables us to define $M$ bijections $\sigma^i:[N]\to \Pi_{\A}^i$ as follows 
	\begin{align}
	&\sigma^{\gamma_k(i)}(k) \triangleq \sigma(Mk-i) .
	\end{align}
	If $A$ and $\sigma$ are also order-preserving then
	it must be that $\sigma^i(1) > \sigma^i(2) > \ldots > \sigma^i(N)$, which means that $\sigma^i$ is a bijection
	from $[N]$ to $ \Pi^i_{\A}$ that agrees with the posterior order bijection $\pi^i_{\A}$. Since $\pi^i_{\A}$ is unique, we conclude that
	$\sigma^i = \pi^i_{\A}$. We have thus obtained the following set equalities:
	\begin{align}
	\{\sigma(Mk-i)\}_{i=0}^{M-1} &= \{\sigma^{\gamma_k(i)}(k)\}_{i=0}^{M-1}\\
	&= \{\sigma^i(k)\}_{i=0}^{M-1}\\
	&= \{\pi^i_{\A}(k)\}_{i=0}^{M-1},
	\end{align}
	and hence $\sigma$ is induced by $A$. 		
\end{proof}

\begin{theorem} \label{thm_gen_np_hard}
	Deciding for general $\sigma$ if it is induced by some partition $\A$ is NP-hard 
\end{theorem}
To prove this, we will show that testing whether $\sigma$ has a posterior-respecting partition, is NP-hard. In light of Lemma~\ref{lemma_gen_feas_bi}, it is in fact a stronger statement.\\

For brevity, before proceeding to prove the theorem, we introduce a more accessible equivalent formulation of the problem. From this point on, we redefine $\Pi$ as $\Pi \triangleq [N]\times[M]$ and each term $\p_v p(k)$ is replaced with a tuple $(k, v+1)$ (notice that the $\pm1$ correction is due to a change from zero-based to one-based indexing). Respectively, we will change the definition of $\sigma$ to an equivalent definition $\sigma:[N]\times[M] \to [N]\times[M]$, which follows the following equivalence relation with the previous definition
\begin{align}
	\sigma(Mj-i+1) = \p_{v-1}p(k) \Longleftrightarrow \sigma(j, i) = (k, v) ,
\end{align}
for all $k,j \in [N]$ and $i,v \in [M]$. Let us define a mapping $\zeta:[N]\times[M]\to \mathbb{Z}_M$, that collects for a given partition $\A$, the index of the posterior set that each $(n, m)$ belongs to. i.e., if $(n, m) \in \Pi^y_{\A}$ then $\zeta(n, m) = y$. Due to Lemma~\ref{lemma_gen_pi_part}
\begin{align}
	\zeta(n, m) = \zeta(n, 1) + m - 1,
\end{align}
where additions from this point on are done modulo $M$, and we define $z_n \triangleq \zeta(n, 1)$. Now, given that $\sigma(k, i) = (n_1, m_1)$ and $\sigma(k,j) = (n_2, m_2)$, (\ref{gen_post_resp_cond}) becomes
\begin{align}
	& \zeta(\sigma(k, i)) \ne \zeta(\sigma(k, j)) \\
	& \Longrightarrow \zeta(n_1, m_1) \ne \zeta(n_2, m_2) \\
	& \Longrightarrow z_{n_1} + m_1 \ne z_{n_2} + m_2 \\
	& \Longrightarrow z_{n_1} - z_{n_2} \ne m_2 - m_1 \label{dmd}.
\end{align}
Any partition $\A$ which is posterior-respecting with respect to $\sigma$ must obey the above condition for all $k \in [N]$ and $i\ne j \in [M]$. Conversely, any assignment to $\{z_n\}_{n=1}^N$ that satisfies the above condition for all $k \in [N]$ and $i\ne j \in [M]$, uniquely defines a posterior-respecting partition. We will show, that testing whether such an assignment exists, is NP-hard, and this is why also asking if there exists a posterior-respecting partition is NP-hard. We call (\ref{dmd}) a {\em Difference Modular Disequation} (DMD). In \cite{jain2008efficient}, Himanshu et al. showed that classifying whether a system of linear modular disequations is satisfiable is NP-hard. They do this, by reducing a 3-SAT problem to a system of linear modular disequations. We will prove that also answering whether a system of DMDs, i.e set of equations of the form $w_i - w_j \not\equiv c_k \pmod{M}$ (for some constants $c_k$) is NP-hard, which is a strengthening of \cite{jain2008efficient}. To prove this we will reduce another well known NP-hard problem, the Not All Equal-3SAT (NAE-3SAT) \cite{moret1988planar} to a system of DMDs.

\begin{lemma} \label{lemma_gen_lmd_np_hard}
	The problem of deciding whether a system of DMDs can be satisfied is NP-hard.
\end{lemma}
\begin{proof}
	We prove the lemma for $M=3$ and discuss the (trivial) extension to $M > 3$ at the end. Let the variables of the NAE-3SAT problem be $x_1, x_2, \ldots, x_n$. For each variable $x_i$ in NAE-3SAT, we introduce two integer variables $w_i$ and $\hat{w}_i$ and we add the following equations to the system:
	\begin{align}
		w_i - s         &\not\equiv 0\\
		\hat{w}_i - s   &\not\equiv 0\\
		w_i - \hat{w}_i &\not\equiv 0,
	\end{align}
	where $s$ is a common variable. We now may define a mapping between the value of $w_i - s \mod 3$ and the value of $x_i$. It is not very important, so we will choose that if the difference is $1$, $x_i$ is false, and if the difference is $2$, $x_i$ is true. Then the value of $\hat{w}_i - s \mod{3}$ represents $\neg{x_i}$. Let $\delta$ be a mapping between literals and their corresponding variable, i.e. $\delta(x_i) = w_i$ and $\delta(\neg{x_i}) = \hat{w}_i$. For each clause $u \vee v \vee w$, we introduce an integer variable $c_i$ and add the following equations to the system:
	\begin{align}
		\delta(u) - c_i &\not\equiv 0\\
		\delta(v) - c_i &\not\equiv 1\\
		\delta(w) - c_i &\not\equiv 2.
	\end{align}
	Given a solution to the system, we use the mapping that was defined above to find the corresponding assignment to $\{x_i\}_{i=1}^{n}$. Notice that for some clause $i$, $c_i$ has no legal value if and only if $\delta(u) = \delta(v) = \delta(w)$, so if there is a solution, it cannot be that all the literals of some clause are equal. The reduction from the given NAE-3SAT to the system of equations is polynomial time. To extend this to $M > 3$, we reduce a NAE-$M$-SAT to a system of DMDs modulo $M$ the same way. Obviously NAE-$M$-SAT for $M > 3$ is still NP-hard. It can be proven for example recursively, by reducing a NAE-$(M-1)$-SAT to NAE-$M$-SAT. Given that $c$ is a clause from a NAE-$(M-1)$-SAT instance, we replace it with $(c \vee x) \wedge (c \vee y)$, where $x$ and $y$ are dummy variables, and the extra clause $x \vee x \ldots \vee x \vee y$ ($x$ appears $M-1$ times) forces $x = \neg y$. Then, for example if $x=0$, the clause $c \vee x$, forces the literals in $c$ to be not all equal to $0$, and the clause $c \vee y$ forces the literals in $c$ to be not all equal to $1$.
\end{proof}

Next, we will show that it is possible to reduce any system of DMDs to a problem of deciding whether for some partition $\sigma$, there is an assignment to $\{z_n\}_{n=1}^{N}$ such that the corresponding partition is posterior-respecting with respect to $\sigma$.

\begin{lemma}\label{lemma_post_resp_np_hard}
	For a general $\sigma$, testing whether there is a posterior-respecting partition is NP-hard.
\end{lemma}
\begin{proof}
	Throughout the proof, we show how to construct $\sigma$, that generates an equivalent system of DMDs for any given system of DMDs, thus showing that finding a posterior-respecting partition is generally at least as hard as solving a system of DMDs. We will prove it for $M = 3$, but the same technique also extends to $M > 3$. Given the following equation $i$
	\begin{align}\label{eq:wk_minus_wl}
		w_k-w_l \not\equiv c_i,
	\end{align}
	we may try a straightforward mapping, and have $w_k$ and $w_l$ to be mapped to $z_k$ and $z_l$ respectively. By having $\sigma(1,1) = (k, 1)$ and $\sigma(1,2)=(l, c_i+1)$, according to (\ref{dmd}), this mapping generates an equivalent equation to (\ref{eq:wk_minus_wl}). However, this naive approach does not scale, since we run into trouble in case $w_k$ appears in another disequation. Recall that because $\sigma$ is a bijection, we cannot use $(k, 1)$ for a different input to $\sigma$. To be able to scale, we introduce a duplication gadget that duplicates the variable $z_k$, i.e., we will add a row to $\zeta$ such that $\zeta(k, *) = \zeta(k', *)$ \footnote{* stands for any index}. We add the following to $\sigma$
	\begin{align}
		& \sigma(1, *) = \{ (k, 1),   (j, 1), (j', 1)\}\\
		& \sigma(2, *) = \{ (k', 2),  (j, 2), (j', 2)\}\\
		& \sigma(3, *) = \{ (k'', 3), (j, 3), (j', 3)\},
	\end{align}
	where $k'$, $k''$ are duplication rows, and $j$, $j'$ are some helper rows that do not correspond to any variable in the original system of equations (we use * because the order between the $\sigma$-siblings does not matter). These rows in $\sigma$ generate the following DMDs:
	\begin{align}
		z_k - z_j       &\not\equiv 0 \\
		z_k - z_j'      &\not\equiv 0 \\
		z_{k'} - z_j    &\not\equiv 0 \\
		z_{k'} - z_{j'} &\not\equiv 0 \\
		z_{k''} - z_j   &\not\equiv 0 \\
		z_{k''} - z_{j'}&\not\equiv 0 \\
		z_j - z_{j'}    &\not\equiv 0.
	\end{align}
	A solution to this system must have $z_j \ne z_{j'}$ and $z_k$, $z_{k'}$ and $z_{k''}$ must be different from both $z_j$ and $z_{j'}$. So it must be that
	\begin{align}
		z_k = z_{k'} =  z_{k''},
	\end{align}
	thus we can use $(k, *)$, $(k', *)$ or $(k'', *)$ interchangeably when constructing $\sigma$. In order to maintain clearer presentation, we will not carry the $k'$ and $k''$, and just assume that we have three copies of $(k, 1)$, $(k, 2)$ and $(k, 3)$, where one of each was already been used for duplication. Now, to represent the disequation we add the following to $\sigma$
	\begin{align}
		& \sigma(4, *) = \{(k, 1),  (l, c_i+1), (i, 1)\} \\
		& \sigma(5, *) = \{(k, 2),  (l, c_i+2), (i, 2)\} \\
		& \sigma(6, *) = \{(k, 3),  (l, c_i+3), (i, 3)\},
	\end{align}
	where $i$ is a helper row that correspond to equation $i$. Indeed all of the above rows in $\sigma$ generate the same DMDs, but we need to place all the $(k, *)$, $(l, *)$ and $(i, *)$ somewhere in $\sigma$. We use the last copy of $(k, *)$ to duplicate and/or represent another disequation. For each disequation, we generate at most a constant number of corresponding duplications, hence showing that there is an assignment to $\{z_n\}_{n=1}^{N}$ that corresponds to a posterior-respecting partition with respect to the $\sigma$ we have constructed, is at least as hard as deciding whether a system of DMDs can be satisfied, and using Lemma~\ref{lemma_gen_lmd_np_hard} it is NP-hard.
\end{proof}
Theorem~\ref{thm_gen_np_hard} is a direct corollary of Lemma~\ref{lemma_gen_feas_bi} and Lemma~\ref{lemma_post_resp_np_hard}. We will conclude with an example of a case where the unconstrained optimum is not achievable and propose a non hermetic method of testing whether it is achievable. 
\begin{example}\label{ex:trenary_ex}
	For $M=3$ and $V \sim [\p_0 = 0.5, \p_1 = 0.3, \p_2 = 0.2]$, then for ${X \sim [0.35\ 0.26\ 0.24\ 0.15]}$, there is no partition $\{A^i\}_{i=0}^{2}$ that achieves the unconstrained optimum. A necessary condition for $\sigma^\downarrow$ to be induced by some partition is that there exists a posterior-respecting partition. We will try to construct such a partition. Without loss of generality, we start by assigning $x_4$ to $A^0$, therefore $\p_0 p(4) \in \Pi_{\A}^0$, $\p_1 p(4) \in \Pi_{\A}^1$ and $\p_2 p(4) \in \Pi_{\A}^2$. Table~\ref{tab:trenary_ex} shows the ternary $\sigma^\downarrow$-siblings. Then, in order to split  $\{\sigma^\downarrow(4, 1),\sigma^\downarrow(4, 2),\sigma^\downarrow(4, 3)\}$ between the posterior sets, we must assign $x_3$ to $A^1$. It is left to the reader to verify that any assignment to $x_1$ and $x_2$ does not end up with a posterior-respecting partition.
\end{example}

\begin{table}
	\captionsetup{justification=centering}
	\begin{center}
		\setlength\tabcolsep{4.9pt}
		\begin{tabular}{|c|c|c|c|}
			\hline
			$\sigma$ & 1 & 2 & 3 \\
			\hline
			1 & $\p_0 p(1) / .175$  & $\p_0 p(2) / .13$ & \cellcolor{yellow!25}$\p_0 p(3) / .12\in \Pi_A^1$\\
			\hline 
			2 & $\p_1 p(1) / .105$  & $\p_1 p(2) / .078$ & \cellcolor{blue!25}$\p_0 p(4) / .075\in \Pi_A^0$\\
			\hline
			3 & \cellcolor{red!25} $\p_1 p(3) / .072\in \Pi_A^2$  & $\p_2 p(1) / .07$ & $\p_2 p(2) / .052$\\
			\hline
			4 & \cellcolor{blue!25}$\p_2 p(3) / .048\in \Pi_A^0$  & \cellcolor{yellow!25}$\p_1 p(4) / .045\in \Pi_A^1$ & \cellcolor{red!25}$\p_2 p(4) / .03\in \Pi_A^2$\\
			\hline
		\end{tabular}
	\end{center}
	\caption{Table of $\sigma^\downarrow$-siblings generated by Example~\ref{ex:trenary_ex}. Each cell contains posterior probability expression/value pair. The coloring refers to the posterior sets.}
	\label{tab:trenary_ex}
\end{table}

\section{Discussion}\label{sec:discussion}
We have shown that the zigzag partition, which amounts to querying whether $X$ has an odd or even index when ordered in descending order of probabilities, is the best question Bob can ask Carole in order to uniformly minimize the expectation of any nondecreasing function of his guessing time, regardless of Carole's lying probability. This result is limited to the case of yes/no questions and a binary symmetric channel from Carole to Bob. Natural extensions of this problem are therefore 1) let Bob ask multiple-choice $M$-ary questions, i.e., to partition $[N]$ into $M$ sets $\{A^i\}_{i=0}^{M-1}$ and ask Carole which one contains $X$, and 2) consider more general channel models for Carole's noisy reply. 

We note that our proof of Theorem~\ref{thrm:main} is almost trivially extended to the $M$-ary case when the channel from Carole to Bob is modulo-additive with a uniform crossover probability, i.e., where Carole answer truthfully with probability $1-\p$, and gives any one of the other $M-1$ incorrect answers with probability $\frac{\p}{M-1}$. This setup reduces to the one discussed on this paper when $M=2$. For arbitrary $M$, the corresponding zigzag partition is the collection of disjoint subsets $\{A^i_{ZZ}\}_{i=0}^{M-1}$ given by   
	\begin{align}
		A^i_{ZZ} \triangleq\left\{k\in [N] : k \equiv i \pmod{M}\right\}.
	\end{align}
This choice is optimal and achieves the corresponding unconstrained optimum (just as~\eqref{eq:best_possible} is achieved in the $M=2$ case). 

Interestingly, our approach does not extend when replacing the special symmetric channel above with a general (symmetric!) modulo-additive channel; in fact, for such channels the unconstrained optimum cannot always be achieved (by any partition), and zigzag is not always optimal. The problem of exactly characterizing the optimal partition or even its performance in this setup appears to be hard. It is thus interesting to examine the applicability of the max-cut / quadratic relaxation approach of~\cite{burin2018reducing} to possibly obtain bounds. It is possible to transform a problem of testing whether there exists a solution to system of DMDs, to a problem of testing if a maximum independent set of a certain graph $G$ is of size $N$, where $N$ is the number of variables in the system. This alternative formulation may hopefully allow the use of graph-theoretic techniques to show interesting properties of specific bijections, such as the unconstrained optimum bijection $\sigma^\downarrow$. We construct $G$ in the following way. For each variable $w_k$, we add a clique of $M$ vertices, indexed $0, 1, \ldots, M-1$. For each equation $w_k - w_l \not\equiv c_i$, we add an edge between the vertex $c_i + m$ of $w_k$ and the vertex $m$ of $w_l$, for all $m \in [M]$. This ensures that if we take two nonadjacent (independent) vertices from these two cliques, then their indices will satisfy the DMD. Thus, if there is an independent set of size $N$ in $G$, then there is a solution to the system of DMDs. Any graph $G$ has the following property \cite{knuth1994sandwich}
\begin{align}
\alpha(G) \le \vartheta(G) \le \bar\chi(G),
\end{align}
where $\alpha(G)$ and $\bar\chi(G)$ are the independence number and the clique partition number, both of which are NP-hard to compute, and $\vartheta(G)$ is the Lov\'asz number which can be computed in polynomial time. For graph instances that represent a system of DMDs, $\alpha(G)$ and $\bar\chi(G)$ are bounded 
\begin{align}
\alpha(G) \le \bar\chi(G) \le N,
\end{align}
so, if $\vartheta(G) < N$, then there is no solution to the system of DMDs. Otherwise, as a result of what we have proven before, it is computationally NP-hard to test whether $\alpha(G) < N$ (this was also proved in \cite{busygin2006np}). This is not necessarily the case for the system of DMDs that is generated by $\sigma^\downarrow$, and it is an open question whether it remains NP-hard to test if the unconstrained optimum is achievable. For example, if the following property is true for the graph $G^\downarrow$ that is generated from $\sigma^\downarrow$
\begin{align}
\alpha(G^\downarrow) < N \Longrightarrow \vartheta(G^\downarrow) < N ,
\end{align}
it the would make the problem of testing whether the unconstrained optimum can be achieved solvable in polynomial time.

It is also interesting to go back to the binary case but consider an asymmetric channel model, i.e., where the crossover probability depends on the input. For this channel, we have derived a quadratic time algorithm for finding the optimal partition for the expected guessing time (not for a general function), in the case of ``sufficiently small'' crossover probabilities, satisfying  for all $i,j\in[N]$ 
	\begin{align}
		\epsilon p(i) &< \bar\delta p(j)\\
		\delta p(i) &< \bar\epsilon p(j), 
	\end{align}
where $\epsilon$ (resp. $\delta$) is the probability of crossing $0 \rightarrow 1$ (resp. $1 \rightarrow 0$). In this case, Bob's optimal strategy regardless of the partition he has used, is to first guess the values of $X$ from the set pointed out by Carole, and only then go over the values in the complement set (according to the posterior order). Given a partition $A$, let  $x_k^{A}$ (resp. $x_k^{\bar A}$) be the posterior order within $A$ (resp. $\bar A$), i.e. $p(x_k^{A}) \ge p(x_{k+1}^{A})$ (resp. $p(x_k^{\bar A}) \ge p(x_{k+1}^{\bar A})$). Then Bob's expected guessing time ($f(k)=k$) is given by
	\begin{align}
		G_A(X) &= \labeledunderbrace{y = 0}{\sum_{k=1}^{\abs{\bar A}} k \cdot \bar\epsilon p(x_k^{\bar A}) + \sum_{k=1}^{\abs{A}} (k + \abs{\bar A}) \cdot \delta p(x_k^A)} \\
		       &+ \labeledunderbrace{y = 1}{\sum_{k=1}^{\abs{A}} k \cdot \bar\delta p(x_k^{A}) + \sum_{k=1}^{\abs{\bar A}} (k + \abs{A}) \cdot \epsilon p(x_k^{\bar A})} \\
		       & = \sum_{k=1}^{\abs{\bar A}} \labeledunderbrace{\bar c_k}{(k + \epsilon\abs{A})} \cdot p(x_k^{\bar A}) + 
		       \sum_{k=1}^{\abs{A}} \labeledunderbrace{c_k}{(k + \delta \abs{\bar A})} \cdot p(x_k^A) \\
		       & = \sum_{k=1}^{\abs{\bar A}} \bar c_k \cdot p(x_k^{\bar A}) + \sum_{k=1}^{\abs{A}} c_k \cdot p(x_k^A).
		       \label{eq:asym_ga}
	\end{align}

Fixing the size of $A$, the coefficients $c_k$ and $\bar c_k$ in~\eqref{eq:asym_ga} are known. Ordering these $N$ coefficients in descending order, denoted by $d^\downarrow_k$, and noting that $p(x_k)$ is a nondecreasing function, we can appeal to Lemma~\ref{lem:rearrange} and obtain 
	\begin{align}
		G_A(X) \ge \sum_{k \in [N]} d^\downarrow_k \cdot p(x_k).
	\end{align}
We can easily achieve this bound by assigning $x_k$ to $A$ if and only if $d^\downarrow_k$ is in $\{c_k\}_{k=1}^{\abs{A}}$. By iterating over the size of $A$, it is possible to find the optimal partition in $O(N^2)$ steps ($N$ evaluations of $G_A(X)$).
	
\bibliography{Minimum_Guesswork_with_an_Unreliable_Oracle}
\bibliographystyle{IEEEtran}

\end{document}